\documentclass[12pt,onecolumn,draftclsnofoot,journal]{IEEEtran}
\usepackage{graphicx}
\usepackage{float}
\usepackage{epstopdf}
\usepackage[cmex10]{amsmath}
\usepackage{array}
\usepackage{cite}
\usepackage{amssymb}
\usepackage{amsfonts}
\usepackage{amsmath}
\usepackage{stackrel}
\usepackage{booktabs,multirow}
\usepackage{arydshln}
\usepackage{amsthm}
\usepackage{listings}
\usepackage{algorithm}
\usepackage{algorithmicx}
\usepackage{algpseudocode}
\usepackage{threeparttable}

\newtheorem{lem}{Lemma}
\newtheorem{rem}{Remark}
\newtheorem{theo}{Theorem}

\newfloat{routine}{htbp}{loa}
\floatname{routine}{Routine}

\makeatletter
\newcommand{\algmargin}{\the\ALG@thistlm}
\makeatother
\newlength{\forwidth}
\algdef{SE}[parFOR]{parFor}{EndparFor}[1]
  {\parbox[t]{\dimexpr\linewidth-\algmargin}{
     \hangindent\forwidth\strut\algorithmicfor\ #1\ \algorithmicdo\strut}}{\algorithmicend\ \algorithmicfor}
\algnewcommand{\parState}[1]{\State
  \parbox[t]{\dimexpr\linewidth-\algmargin}{\strut #1\strut}}

\newlength{\ifwidth}
\algdef{SE}[parIF]{parIf}{EndparIf}[1]
  {\parbox[t]{\dimexpr\linewidth-\algmargin}{
     \hangindent\ifwidth\strut\algorithmicif\ #1\ \algorithmicdo\strut}}{\algorithmicend\ \algorithmicif}

\begin{document}

\title{ \Large 
Cospectral Bipartite Graphs with the Same Degree Sequences but with Different Number of Large Cycles}
\author{\IEEEauthorblockN{Ali Dehghan  and Amir H. Banihashemi}
	\thanks{This paper was presented
		in part at ISTC 2018, Hong Kong. This research was supported by NSERC Discovery Grant 217239-2013-RGPIN.}
	\IEEEauthorblockA{\\\small Department of Systems and Computer Engineering, Carleton University, Ottawa, Ontario, Canada}
}

\maketitle


\begin{abstract}
Finding the multiplicity of cycles in bipartite graphs is a fundamental problem of interest in many fields including the analysis and design of low-density parity-check (LDPC) codes.
Recently, Blake and Lin computed the number of shortest cycles ($g$-cycles,  where $g$ is the girth of the graph) in a bi-regular bipartite graph, in terms of the degree sequences and the spectrum (eigenvalues of the adjacency matrix) of the graph [{\em IEEE Trans. Inform. Theory 64(10):6526--6535, 2018}]. This result was subsequently extended in [{\em IEEE Trans. Inform. Theory, accepted for publication, Dec. 2018}] to cycles of length $g+2, \ldots, 2g-2$, in bi-regular bipartite graphs, as well as $4$-cycles and $6$-cycles in irregular and half-regular bipartite graphs, with $g \geq 4$ and $g \geq 6$, respectively. In this paper, we complement these positive results with negative results demonstrating that the information of the degree sequences and the spectrum of a bipartite graph is, in general, insufficient 
to count (a) the $i$-cycles, $i \geq 2g$, in bi-regular graphs, (b) the $i$-cycles for any $i > g$, regardless of the value of $g$, and $g$-cycles for $g \geq 6$, in irregular graphs, and (c) the $i$-cycles for any $i > g$, regardless of the value of $g$, and $g$-cycles for $g \geq 8$, in half-regular graphs. To obtain these results, we construct counter-examples using the Godsil-McKay switching.

\begin{flushleft}

\end{flushleft}

\end{abstract}

\section{introduction}

Bipartite graphs are commonly used in science and engineering to represent systems, where the nodes on one side of the bipartition represent the {\em variables}, and the nodes on the other side represent local {\em constraints}, each involving its adjacent variables, see, e.g.,~\cite{MR1820474}. A well-known example is the Tanner graph representation~\cite{MR650686} of low-density parity-check (LDPC) codes~\cite{book}. 
In graph representations, cycles and degree sequences often play an important role in determining the performance of the system. For example, in the case of LDPC codes, the performance of iterative message-passing algorithms, that are used in practice for the decoding, depends highly on the cycle distribution and the degree sequences of the underlying Tanner graph~\cite{MR1820479, MR2810002, xiao2009error, hu2005regular, MR3071345, asvadi2011lowering, MR2991821, MR3252383, HB-IT1, HB-IT2}. For the purpose of analysis and design of systems and codes, it is thus important to know the number of cycles of different length in the corresponding bipartite graphs, and the relationships that may exist between the cycle distribution and the degree sequences of the graph.  

The connection between the  performance of LDPC codes and cycles of the Tanner graph has motivated much research on the study of the cycle distribution and 
the counting of cycles in bipartite graphs, see, e.g.,~\cite{blake2017short},~\cite{dehghan2016new},~\cite{eigenvalue},~\cite{halford2006algorithm},~\cite{karimi2013message}.
Counting cycles of a given length, even in bipartite graphs, is known to be NP-hard~\cite{MR1405031}. 
In \cite{karimi2013message}, Karimi and Banihashemi  presented an efficient message-passing algorithm to count the number of cycles 
of length less than $2g$, in a general graph, where $g$ is the girth of the graph. The distribution of cycles in different ensembles of bipartite graphs was studied in~\cite{dehghan2016new},
where it was shown that for random ensembles of bipartite graphs, the multiplicities of cycles of different lengths have independent Poisson distributions with the expected values only a function of the cycle length and the degree distribution (and independent of the size of the graph). More recently, Blake and Lin~\cite{blake2017short} presented 
a formula to compute the multiplicity of cycles of length $g$ in bi-regular bipartite graphs as a function of the spectrum (eigenvalues of the adjacency matrix of the graph) and degree sequences of the graph. 
This result was subsequently extended in \cite{eigenvalue} to compute the number of cycles of length $g+2, \ldots, 2g-2$, in bi-regular bipartite graphs, 
as well as the number of $4$-cycles and $6$-cycles in irregular and half-regular bipartite graphs, with $g \geq 4$ and $g \geq 6$, respectively. 
It is noteworthy that, while the majority of techniques developed in the literature for counting cycles are algorithmic, 
the results in \cite{blake2017short} and \cite{eigenvalue} are presented as closed-form formulas. 

In relation to the results of \cite{blake2017short} and \cite{eigenvalue}, that use 
the spectrum $\{\lambda_i\}$ of a bipartite graph as part of the required information to derive the cycle multiplicities, we note that, in general, determining the properties of a 
graph from its spectrum is an active area of research in graph theory. For some examples, see \cite{MR2112881, MR2979293, MR3854106, MR2600481, MR2811143}.
It is known that $\sum_{i}\lambda_i^j$ for $j = 1, 2$, and $3$ is equal to $0$, the number of edges in the graph, and six times  the number of $3$-cycles of the graph, respectively. 
It is, however, not possible to extend the last result to cycles of length larger than three. (For example,  
the complete bipartite graph with one and four nodes on the two sides of the bipartition, and the union of a $4$-cycle and a single node are two bipartite graphs 
with the same  spectrum $\{-2,0, 0, 0,2\}$, but with different number of $4$-cycles.) The results of \cite{blake2017short} and \cite{eigenvalue} make this extension possible 
but with the extra information about the degree sequences of the graph.

In this paper, we complement the results of \cite{blake2017short} and \cite{eigenvalue} by demonstrating, through counter-examples, that the information 
of the degree sequences and the spectrum of a bipartite graph is, in general, insufficient to count (a) the $i$-cycles with $i \geq 2g$ in bi-regular graphs, (b) the $i$-cycles for any $i > g$, regardless of the value of $g$, and $g$-cycles for $g \geq 6$, in irregular graphs, and (c) the $i$-cycles for any $i > g$, regardless of the value of $g$, and $g$-cycles for $g \geq 8$, in half-regular graphs.
To construct our counter-examples, we use the Godsil-McKay switching \cite{MR730486}, and prove that the application of such switches to bi-regular bipartite graphs preserves the degree sequences and the spectrum of the graph.

We note that in graph theory, the Godsil-McKay switching is a well-known tool to construct cospectral graphs. For example, Bl\'{a}zsik {\it et al.} \cite{MR3291884} used the switching to construct 
two cospectral regular graphs such that one has a perfect matching while the other does not have any perfect matching. For more applications, see \cite{MR3825701, MR3742840, MR3654202}.

The summary of the results regarding the possibility of computing the number of cycles of different length in different types of bipartite graphs with different girth using only the spectrum and the degree sequences of the graph is presented in Table \ref{TXY1}. In this table, the notation ``P" (``IP'') is used to mean that it is possible (impossible), in general, to 
find the multiplicity of cycles of a given length in a graph from the spectrum and the degree sequences of the graph. 

\begin{table}[ht]
	\caption{A summary of the results on the possibility of counting cycles of length $i$ in bi-regular, half-regular and irregular bipartite graphs with girth $g$ using only the spectrum and the degree sequences of the graph.  (Notations ``P"  and ``IP'' are used for ``possible'' and ``impossible,'' respectively.)}
	\begin{center}
		\scalebox{1}{
			\begin{tabular}{ |c|c||c|c|c|  }
				\hline
				&              & $i=g$ & $ g+2 \leq i \leq  2g-2$ &  $  2g \leq i$ \\ \hline \hline
				
				\multirow{1}{*}{Bi-regular}
				& $g\geq 4$    & P \cite{blake2017short}   & P  \cite{eigenvalue}   & IP (Section \ref{sec3}) \\ \hline \hline

				\multirow{3}{*}{Half-regular}
				& $g = 4   $   & P  \cite{eigenvalue}   & IP  (Subsection \ref{bnh})  & IP (Subsection \ref{bnh}) \\
				& $g = 6$      & P  \cite{eigenvalue}   & IP (Section \ref{sec36})   & IP (Section \ref{sec36})  \\
				& $g\geq 8$    &IP   (Section \ref{sec36})   & IP  (Section \ref{sec36})   & IP (Section \ref{sec36})  \\ \hline \hline
				
				\multirow{3}{*}{Irregular}
				& $g=4$        & P   \cite{eigenvalue}  & IP (Subsection \ref{bnh})   & IP (Subsection \ref{bnh})  \\
				& $g=    6$    &IP   (Subsection \ref{subsec33})  & IP  (Subsection \ref{subsec33})  & IP  (Subsection \ref{subsec33})\\ 
				& $g\geq 8$    &IP   (Subsection \ref{subsec22})  & IP   (Subsection \ref{subsec22}) & IP (Subsection \ref{subsec22}) \\ \hline
				
			\end{tabular}
		}
	\end{center}
	\label{TXY1}
\end{table}

The organization of the rest of the paper is as follows: In Section~\ref{sec1}, we present some definitions and notations. 
Next, in Section~\ref{sec3}, we construct two  bi-regular bipartite graphs such that they have the same spectrum, degree sequences and girth, 
but different number of $i$-cycles for $i \geq 2g$. This  demonstrates that, in general, it is not possible to determine the number of $i$-cycles for $i \geq 2g$ in a bi-regular bipartite graph as a function of only the spectrum and the degree sequences of the graph.
In Section~\ref{sec35}, we study the possibility of computing the multiplicity of short cycles of irregular bipartite graphs using only the spectrum and degree sequences, 
and demonstrate through some graph constructions that the answer is generally negative, except for the case of $4$-cycles in graphs with $g \geq 4$ (the equation for the 
multiplicity of $4$-cycles was derived in \cite{eigenvalue} as a function of graph spectrum and its degree sequences). 
In Section~\ref{sec36}, we continue our study of computing the multiplicity of short cycles in half-regular bipartite graphs, and show that for all girths and cycle sizes, with the exception
of $6$-cycles in graphs with $g \geq 6$ (and $4$-cycles in graphs with $g \geq 4$), the information of only the spectrum and degree sequences is insufficient to count the cycles. 
The paper is concluded with some remarks in Section~\ref{sec5}.

\section{Definitions and notations}
\label{sec1}

A  graph $G $ is defined as a set of vertices or nodes $V(G)$ and a set of edges $E(G)$, where $E(G)$ is a subset of
the pairs $\{\{v,u\}: v,u\in V(G) , v\neq u\}$. The shorthands $V$ and $E$ are used if there is no ambiguity about the graph. An edge $e \in E$ with endpoints $u \in V$ and $w \in V$ is denoted by $\{u,w\}$, or by $uw$ or $wu$, in brief. Throughout this work, we consider undirected graphs with
no loop or parallel edges (i.e., simple graphs).

A {\it walk} of length $k$ in the graph $G$ is a sequence of nodes
$v_1, v_2, \ldots , v_{k+1}$ in $V$ such that $\{v_i, v_{i+1}\} \in E$, for all $i \in \{1, \ldots , k\}$.  A walk is a {\it path} if all the nodes $v_1, v_2, \ldots , v_k$ are distinct. A walk is called a
{\it closed walk}  if the two end nodes are identical, i.e.,
if $v_1 = v_{k+1}$. Under the same condition, a path is called a {\it cycle}. 

The {\em length} of a walk, path or cycle is the number of its edges. We  use the notation $P_n$ to denote a path with $n$ nodes. We denote cycles of length $k$, also referred to as $k$-cycles, by $C_k$.  The length of the shortest cycle(s) in a graph is called {\em girth} and is denoted by $g$.

A graph $G=(V,E)$ is called {\it bipartite}, if the node set $V$ can be
partitioned into two disjoint subsets $U$ and $W$, i.e., $V = U \cup W \text{ and } U \cap W =\emptyset $, such that every edge in $E$ connects a node
from $U$ to a node from $W$. A graph is bipartite if and only if the lengths of all its cycles are even.
Tanner graphs of LDPC codes are bipartite graphs, in which $U$ and $W$ are referred to as {\it variable nodes} and {\it check nodes}, respectively. 
Parameters $n$ and $m$ in this case are used to denote $|U|$ and $|W|$, respectively. Parameter $n$ is the code's block length and the code rate 
$R$ satisfies $R \geq 1- (m/n)$. 

A  graph is called {\em complete}
if every node is connected to all the other nodes. We use the notation $K_a$ for a complete graph with $a$ nodes. A bipartite graph $G(U \cup W, E)$ is called {\em complete}, and is denoted by $K_{|U|,|W|}$, if every node in $U$ is connected to every node in $W$. 

The number of edges incident to a node $v$ is called the {\em degree} of $v$, and is denoted by $d(v)$. The {\it degree sequences} of a bipartite graph $G$ are defined as the two monotonic non-increasing sequences of the node degrees on the two sides of the graph. 
For example, the complete bipartite graph $K_{4,3}$ has degree sequences $(3,3,3,3)$ and $(4,4,4)$. The degree sequences also contain the 
information about the number of nodes on each side of the graph.
A bipartite graph $G = (U\cup W,E)$ is called {\it bi-regular}, if all the nodes on the same side of the bipartition have the same degree,
i.e., if all the nodes in $U$ have the same degree $d_u$ and all the nodes in $W$ have the same degree $d_w$. We also call such graphs $(d_u,d_w)$-regular graphs.
It is clear that, for a bi-regular graph, $|U|d_u=|W|d_w=|E(G)|$. A bipartite graph is called {\em half-regular}, if all the nodes on one side of the bipartition have the same degree.
A half-regular Tanner graph can be either variable-regular or check-regular. 
A bipartite graph that is not bi-regular is called {\it irregular}. With this definition, half-regular graphs are a special case of irregular graphs.

A graph $G$ is {\it connected}, if there is a path between any two nodes of $G$. If the graph $G$ is not connected, we say that it is disconnected. A {\it connected  component}  of  a  graph  is  a  connected
subgraph such that there are no edges between nodes of the subgraph and nodes of the rest of the graph.

The {\it adjacency matrix} of a graph $G$ is the matrix $A = [a_{ij}]$, where $a_{ij}$ is the number of edges connecting the node $i$ to the node
$j$ for all $i, j\in V$. The matrix $A$ is symmetric and since we have assumed that $G$ has no parallel edges or loops, $a_{ij}\in\{0, 1\}$,
for all $i, j\in V$, and $a_{ii} = 0$, for all $i \in V$. The set of the eigenvalues $\{\lambda_i\}$ of $A$ is called the {\em spectrum} of the graph. 
It is well-known that the spectrum of a disconnected graph is  the disjoint union of the spectra of its components \cite{MR2022290}.
One important property of the adjacency matrix is that the number of walks
between any two nodes of the graph can be determined using the powers of this matrix. More precisely, the entry in
the $i^{\text{th}}$ row and the $j^{\text{th}}$ column of $A^k$, $[A^k]_{ij}$ , is the number of walks of length $k$ between nodes $i$ and $j$. In particular, $[A^k]_{ii}$
is the number of closed walks of length $k$ containing node $i$. The total number of closed walks of length $k$ in $G$ is thus $tr(A^k)$, where $tr(\cdot)$ is the trace of a matrix.
Since $tr(A^k)= \sum_{i=1}^{|V|}\lambda_i^k$, it follows that the multiplicity of closed walks of different length in a graph can be obtained using the spectrum of the graph.
It is also known that $\sum_{i}\lambda_i(G)^j$, for $j=1, 2, 3$, is equal to $0$, $|E(G)|$, and $6 \times N_3(G)$, respectively, where $N_3(G)$ is the number of $3$-cycles in $G$. 
It is, however, not possible to extend the last result to cycles of length larger than three, and find the multiplicity of such cycles as a function of only the spectrum. 
For example,  the complete bipartite graph $K_{1,4}$ and the graph $C _4 \cup K_1$ (the union of a $4$-cycle and a single node) 
are two bipartite graphs with the same  spectrum $\{-2,0, 0, 0,2\}$, but with different number of $4$-cycles.

To devise our counter-examples, we often use cycles and paths. The path graph $P_n$ has the following spectrum:
\begin{equation}
2\cos \Big(\dfrac{\pi j}{n+1}\Big),\,\,\, j=1,\ldots,n\:,
\label{path}
\end{equation}
and the spectrum of a cycle of length $n$, $C_n$, is as follows:
\begin{equation}
2\cos \Big(\dfrac{2\pi j}{n }\Big),\,\,\, j=0,\ldots,n-1\:.
\label{cycle}
\end{equation}

In general, the spectrum of a graph does not uniquely determine the graph. Two graphs are called {\it cospectral} or {\it isospectral} if they have the same spectrum.
On the other hand, there are graphs that are known to be uniquely determined by their spectrum.
Two examples are the complete graph $K_n$, and the complete bipartite graph $K_{n,n}$ \cite{MR2022290}.

\section{Counting large cycles in bi-regular bipartite graphs}
\label{sec3}

In this section, we demonstrate that the knowledge of spectrum and degree sequences of a bi-regular bipartite  graph is not in general sufficient to determine the multiplicity of cycles of length $2g$ and larger.
We start by providing a counter-example of two regular bipartite graphs whose spectrum, degree sequences and girth are identical but have different number of cycles of length $2g$ and larger. 
To construct this counter-example, we use the concept of {\it switching} in graphs and in particular, {\it Godsil-McKay switching}~\cite{MR730486}. The latter is a graph transformation 
that maintains the spectrum of the graph. We also prove that  Godsil-McKay switching, in general, maintains the degree sequences of a bi-regular bipartite graph and can thus be used to construct cospectral bi-regular bipartite graphs 
with similar degree sequences, but different cycle distributions for cycle lengths larger than $2g-2$.

\begin{theo}\label{Th5} [Godsil-McKay switching \cite{MR730486}]
Let $G$ be a graph and let $\{ X_1, \ldots, X_{\ell}, Y\}$ be a partition of the node set $V (G)$ of $G$.
Suppose that for every node $y\in Y$, and every $i\in \{1,\ldots,\ell\}$, the node $y$ has either $0$, $
\frac{1}{2} |X_i|$ or $|X_i|$ neighbors
in $X_i$. Moreover, suppose that for each  $i,j\in \{1,\ldots,\ell\}$ ($i$ and $j$ can be equal), all the nodes in $ X_i$ have the same number of
neighbors in $X_j$. Construct a new graph $G'$ as follows: For each $y\in  Y$ and $i\in \{1,\ldots,\ell\}$ such that $y$
has $\frac{1}{2} |X_i|$ neighbors in $X_i$, delete the corresponding $\frac{1}{2} |X_i|$ edges and join $y$ instead to the $\frac{1}{2} |X_i|$ other
nodes in $X_i$. Then, the graphs $G$ and $G'$ are cospectral.
\end{theo}

In the above process, the node partition $\{ X_1, \ldots, X_{\ell}, Y\}$ is called a {\it Godsil-McKay switching partition}.


In the following, we construct a 
$3$-regular bipartite graph $G$ with girth six. We then use Theorem~\ref{Th5} to convert $G$ to $G'$, such that $G'$ is also $3$-regular and bipartite, and $G$ and $G'$ are cospectral. 
In our construction of $G$, we use the Heawood graph \cite{MR0411988}, shown in Fig. \ref{graphBXX}. (In the rest of the paper, to make the identification of the nodes on each side of the bipartition easier, we sometimes use black and white colors to distinguish them.) The Heawood graph is a $3$-regular bipartite graph with girth six.

\begin{figure}[ht]
\begin{center}
\includegraphics[scale=.3]{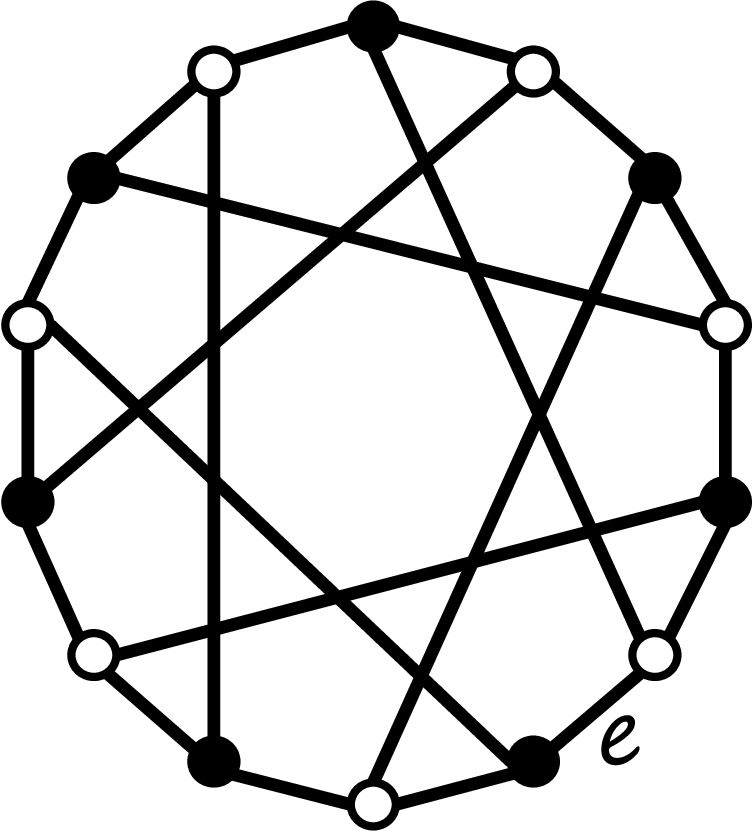}
\caption{The Heawood graph.
} \label{graphBXX}
\end{center}
\end{figure}

\underline{Construction of $G$}:
Consider two disjoint cycles of length $6$ and $18$ with  nodes $ d_1, d_2, \ldots, d_6$, and $a_1, a_2, \ldots, a_6, b_1, \ldots, b_6, c_1, \ldots, c_6$, 
respectively. Add to the graph twelve nodes $v_1,u_1,\ldots,v_6, u_6$, and for each $i\in \{1,2,\ldots,6\}$, 
add the edges $v_id_i, v_ia_i, u_ib_i, u_ic_i$. Further add to the graph nodes $v',v'',u',u''$, and edges $v'v_1, v'v_3, v'v_5, v''v_2, v''v_4, v''v_6, u'u_3, u'u_5,  u''u_4, u''u_6, u'u''$. Now, add a copy of the Heawood graph, and remove one of its circumferential edges such as $e=zz'$ (see Fig. \ref{graphBXX}). Finally, add the edges $u_1z$ and $u_2z'$ to the graph. The resulting graph is a $3$-regular  bipartite graph with girth six. We call this graph $G$. See Fig. \ref{graphB3}.

\begin{figure}[ht]
\begin{center}
\includegraphics[scale=.30]{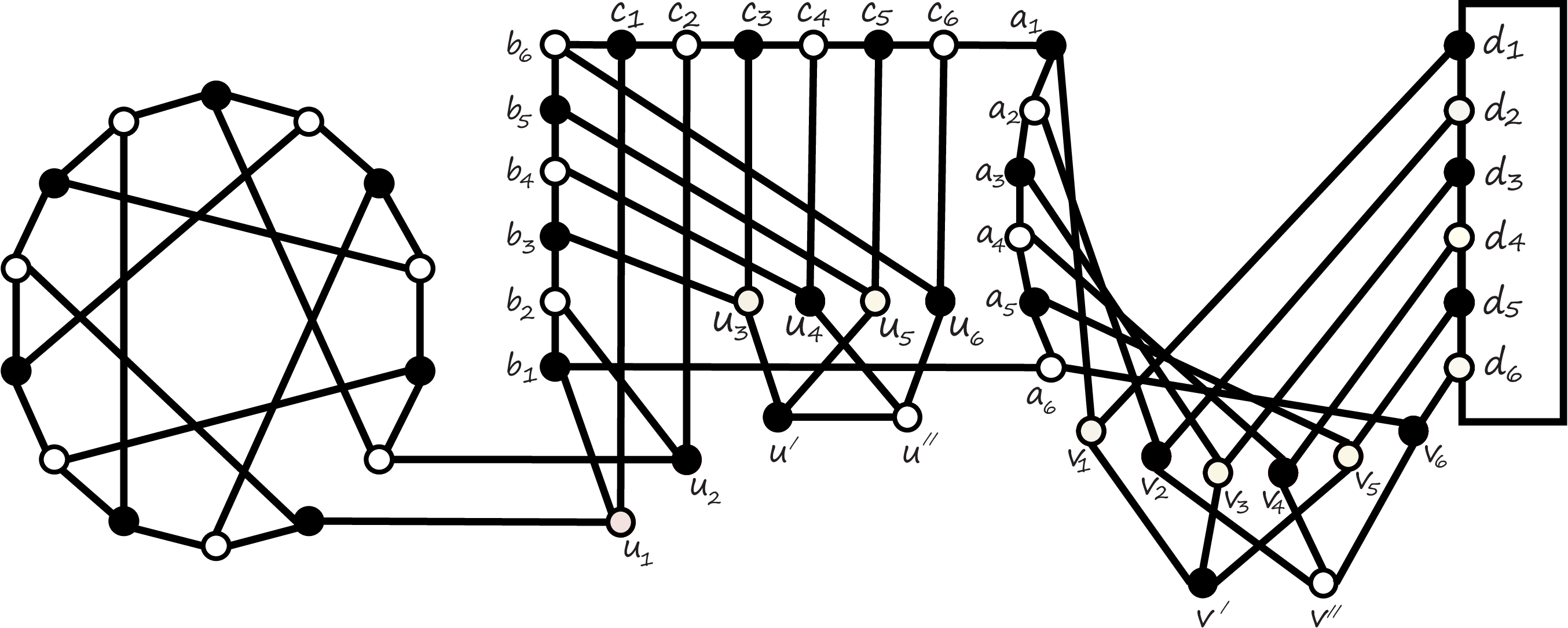}
\caption{The $3$-regular bipartite graph $G$.
} \label{graphB3}
\end{center}
\end{figure}

\underline{Construction of $G'$}: We use Theorem \ref{Th5}, and construct $G'$ from $G$. Let $\ell=6$, and for each $i$, $1\leq i \leq 6$, let $X_i=\{ a_i,b_i,c_i,d_i\}$.
Also, Let $Y= V(G)\setminus \cup_{i=1}^6 X_i$.  It can be seen that for every node $y\in Y$, and every $i\in \{1,\ldots,6\}$, node $y$ has either $0$  or $2$ neighbors
in $X_i$ (note that for each $i$, $|X_i|=4$). Also, for each pair $i,j\in \{1,\ldots,\ell\}$, all the nodes in $X_i$ have the same number of
neighbors in $X_j$ (see Table \ref{V2T1}). 
Consequently, the partitioning has all the properties of Theorem \ref{Th5}. We can thus apply Godsil-McKay switching, and obtain $G'$. (See Fig.~\ref{graphB3N}.)

\begin{figure}[ht]
	\begin{center}
		\includegraphics[scale=.30]{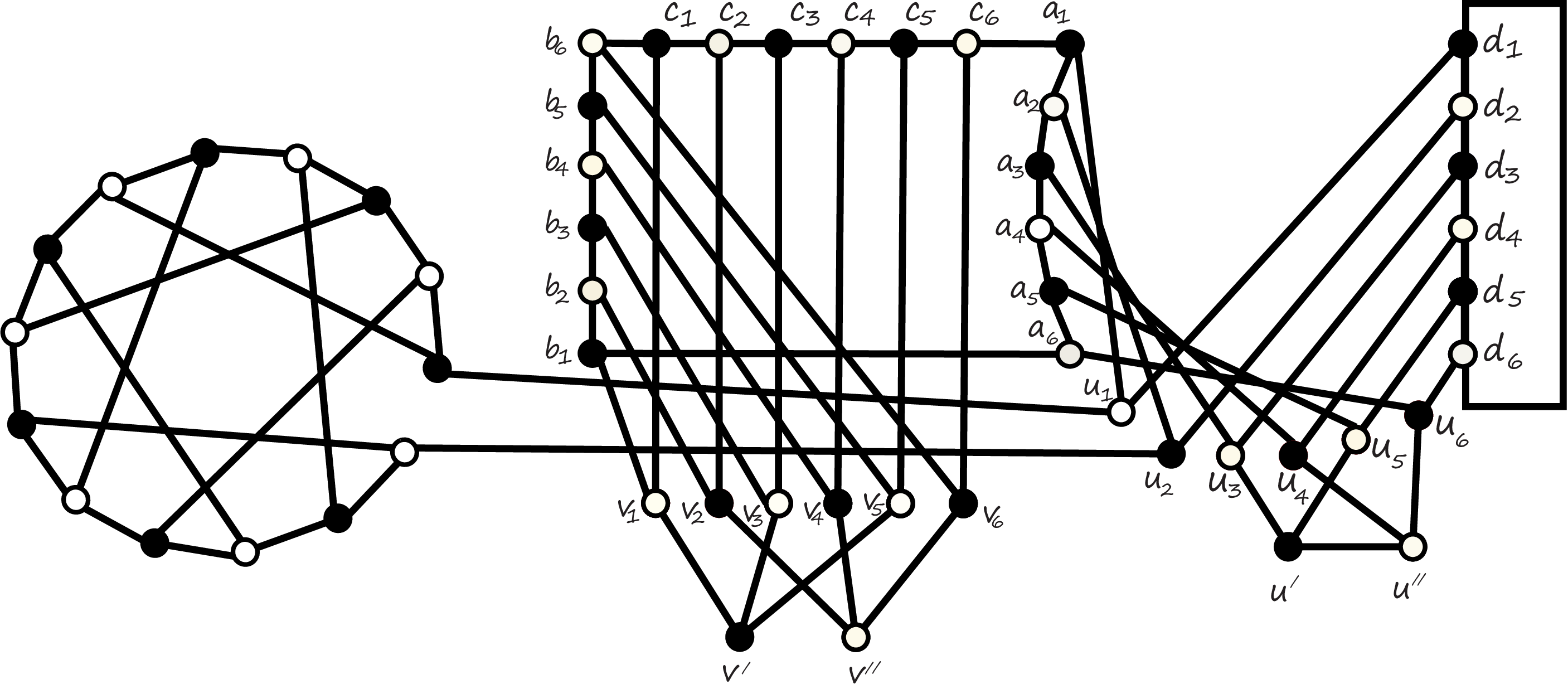}
		\caption{The $3$-regular bipartite graph $G'$, obtained by the application of Godsil-McKay switching to the graph $G$ in Fig.~\ref{graphB3}.
		} \label{graphB3N}
	\end{center}
\end{figure}

\begin{table}[ht]
\caption{The $(i,j)$ entry of the table shows the number of neighbors that an arbitrary node $v \in X_i$ has in the set $X_j$ (for Godsil-McKay switching partition of $G$ shown in Fig.~\ref{graphB3}).}
\begin{center}
\scalebox{1}{
\begin{tabular}{ |c||c|c|c|c|c|c| }
\hline
                               &  $X_1$  & $X_2$  & $X_3$ &  $X_4$  & $X_5$ & $X_6$    \\
\hline
\hline
$X_1$                          & 0       &1       & 0     & 0       &0     & 1       \\
\hline
$X_2$                          & 1       &0       & 1     & 0       &0     & 0       \\
\hline
$X_3$                          & 0       &1       & 0     & 1       &0     & 0       \\
\hline
$X_4$                          & 0       &0      & 1     &  0       &1     & 0       \\
\hline
$X_5$                          & 0       &0       & 0    &1         &0     & 1       \\
\hline
$X_6$                          & 1       &0       & 0     &  0      &1    & 0      \\
\hline
\end{tabular}
}
\end{center}
\label{V2T1}
\end{table}

Both $G$ and $G'$ are $3$-regular bipartite graphs and based on Theorem~\ref{Th5}, both have the same spectrum. 
In Table \ref{TT1}, we have listed the number of cycles of length $6$ up to $22$, for both graphs.\footnote{The cycles are counted using a Matlab program by Jeff Howbert \cite{Jeff}. This program counts all cycles in a simple undirected graph up to a specified size limit, using a backtracking algorithm.}
As expected from the results presented in \cite{eigenvalue}, both graphs have the same cycle distribution for cycles of length up to $2g-2=10$. 
From the table, however, it can be seen that the multiplicities of cycles of length $2g=12$ and larger are different in these graphs.

\begin{table}[ht]
\caption{Multiplicities of cycles of length $6$ up to $22$ in $G$ and $G'$}
\begin{center}
\scalebox{1}{
\begin{tabular}{ |c ||c|c|c|c|c|c|c|c|c|  }
\hline
Graph&  6-cycles &  8-cycles &  10-cycles &  12-cycles &  14-cycles &  16-cycles &  18-cycles &  20-cycles &  22-cycles \\ \hline
$G$  &  51       &  54       &  186       &  212       &  460       &  659       &  1609      &  4038      &  11132  \\ \hline
$G'$ &  51       &  54       &  186       &  213       &  458       &  669       &  1576      &  4090      &  10977 \\
\hline
\end{tabular}
}
\end{center}
\label{TT1}
\end{table}

Although, the example just provided was for regular bipartite graphs, one can use the following theorem  to construct cospectral $(d_u,d_w)$-regular bipartite graphs with $d_u \neq d_w$, whose $i$-cycle multiplicities are different for $i \geq 2g$. 


\begin{theo}\label{V2Th2}
Let $G$ be a bi-regular bipartite graph, and suppose that Godsil-McKay switching is used to convert $G$ into $G'$. Then, the graph $G'$ is also bi-regular and both graphs have the same degree sequences.
\end{theo}

\begin{proof}
For the proof, we first discuss some of the properties of a Godsil-McKay switching partition of a bi-regular bipartite graph.
Let $G=(U\cup W, E)$ be a   bi-regular graph in which all the nodes in $U$ have the same degree $d_u$ and all the nodes in $W$ have the same degree $d_w$.
Let $\{ X_1, \ldots, X_{\ell}, Y\}$ be a  Godsil-McKay switching partition for the nodes of $G$.
For each $i$, we say that the set of nodes $X_i$ is of Type 1 (Type 2), if all nodes of $X_i$ are in $U$ ($W$). 
Otherwise, we say that $X_i$ is of Type 3 (if some nodes of $X_i$ are in $U$ and some others are in $W$). 
Let $X_i$ be a set of Type 3. Partition $X_i$ into two parts $X_i^1$ and $X_i^2$, where $X_i^1$ is the subset of nodes of $X_i$ that are in $U$, and thus $X_i^2$ contains the nodes of $X_i$ that are in $W$. 
Therefore, $|X_i|=|X_i^1|+|X_i^2|$. If $X_i$ is of Type 3, we say it is of Type 3.1, if $|X_i^1|=|X_i^2|$. Otherwise, we say that $X_i$ is of Type 3.2. We then have the following properties for partition sets of different types.

\begin{lem}
	Any Godsil-McKay switching partition $\{ X_1, \ldots, X_{\ell}, Y\}$ of the nodes of a $(d_u,d_w)$-regular bipartite graph $G=(U\cup W, E)$ has the following properties: 
	\begin{itemize}
		\item[P1.] There is no connection (edge) between the nodes of a Type-3 set and the nodes of a Type-1 or Type-2 set.
		\item[P2.] There is no connection between the nodes of a Type-3.1 set and the nodes of a Type-3.2 set.
		\item[P3.] Let $ X_i$ and $X_j$ be two sets of Type 3.2.  Assume that the nodes in $ X_i$ have at least one neighbor in $ X_j$. Then, if $|X_i^1|>|X_i^2|$, we have $|X_j^1|<|X_j^2|$, and if $|X_i^1|<|X_i^2|$, we have $|X_j^1|>|X_j^2|$.
		\item[P4.] Let $X_i$ be a set of Type 3. If a node $y\in Y$ has a neighbor in $X_i$, then $y$ is adjacent with $\frac{|X_i|}{2}$ nodes of $X_i^1$ or $y$ is adjacent with
		$\frac{|X_i|}{2}$ nodes of $X_i^2$ ($y$ cannot have neighbors in both $X_i^1$ and $X_i^2$).
		\item[P5.] For each $X_i$, each node in $X_i$ is connected to the same number of nodes in $\cup_{j=1}^{\ell} X_j$.
	\end{itemize}
\end{lem}
\begin{proof}
	P1. Let $ X_i$ be a set of Type 3 and  $ X_j$ be a set of Type 2 or Type 1. Since the graph is bipartite, some of the nodes of $X_i$ cannot have any connection to the nodes of $X_j$.
	Moreover, all the nodes of $X_i$ must have the same number of neighbors in $X_j$. This number thus must be zero. P2.  Let $ X_i$ be a set of Type 3.2 and  $ X_j$ be a set of Type 3.1.
	Let $|E'|$ be the number of edges between $X_i^1$ and $X_j^2$, and $|E''|$ be the number of edges between $X_i^2$ and $X_j^1$.  Since $ X_j$ is a set of Type 3.1, and every node in $X_j$ has the same number of neighbors in $X_i$,
	we have $|E'|=|E''|$. On the other hand, since  $X_i$ is a set of Type 3.2 and every node in $X_i$ has the same number of neighbors in $X_j$, we have $|E'|\neq|E''|$,  which is a contradiction. The proofs for P3-P5 are straightforward.
\end{proof}

We now prove Theorem~\ref{V2Th2}.
Consider the application of the Godsil-McKay switching to convert a $(d_u,d_w)$-regular bipartite graph $G=(U\cup W, E)$ into the graph $G'$. 
The nodes of $G'$ can be partitioned into two sets $U'$ and $W'$ according to the following rules:\\
\underline{\bf Rule 1.} For each node $v\in Y$ in the graph $G$, assign the corresponding node $v$ in $G'$ to $U'$ ($W'$) if $v$ in $G$ is in $U$ ($W$).\\
\underline{\bf Rule 2.} For each $i$, if $X_i$ is of Type 1 or Type 2, then for each node $v$ in  $X_i$ in the graph $G$, assign the corresponding node $v$ in $G'$ to $U'$ ($W'$) if $v$ in $G$ is in $U$ ($W$).\\
\underline{\bf Rule 3.} For each $i$, if $X_i$ is of Type 3, then for each node $v$ in  $X_i$ in the graph $G$, assign the corresponding node $v$ in $G'$ to $U'$ ($W'$) if $v$ in $G$ is in $ W$ ($U$).

Now, we show that $G'=(U'\cup W', E')$ is a bi-regular graph in which all the nodes in
$U'$ have the same degree $d_u$ and all the nodes in $W'$ have the same degree $d_w$. To show this, we examine the degrees of different partition sets $Y$ and $X_i$'s. For the latter sets, 
the examination is based on the type of the set. 

(i) Set $Y$: It is clear that the Godsil-McKay switching does not change the degree of any node in $Y$, and by Rule 1, those nodes in $Y$ with degree $d_u$ ($d_w$) are in $U'$ ($W'$).

(ii) Type-1 or Type-2 $X_i$: Let $X_i$ be a set of Type 1. If there is no node  $y \in Y$ such that $y$ is adjacent with $\frac{|X_i|}{2}$ nodes of $X_i$, then the Godsil-McKay switching does not change the degree of any node in $X_i$. 
Now, assume that there is a node  $y\in Y$ such that $y$ is adjacent to $\frac{|X_i|}{2}$ nodes of $X_i$. Let $Y_i \subset Y$ be a subset of nodes such that for each node $y\in Y_i$, the node $y$ is adjacent to $\frac{|X_i|}{2}$ nodes of $X_i$. 
Considering that all nodes in $X_i$ have the same degree $d_u$, by using P5, we conclude that
all the nodes in $X_i$ have the same number of neighbors in $Y_i$. Call this number $\gamma$. By counting the number of edges $\eta$ between $X_i$ and $Y_i$, we find that $\eta=\gamma|X_i|$. 
On the other hand, $\eta = |Y_i| |X_i|/2$.  
Thus, $|Y_i|=2\gamma$. This implies that each node $v \in X_i$ is adjacent to half of the nodes in $Y_i$ ($\gamma$ of them), and has no connection to the other half. 
Consequently, the Godsil-McKay switching does not change the degree of any node in $X_i$. This together with Rule 2 shows that each node in any Type-1 set in $U'$ has degree $d_u$.
The proof for a Type-2 set is similar.

(iii) Let $X_i$ be a set of Type 3.1. By P1 and P2, all the connections to $X_i$ are from $Y$ and Type-3.1 sets. Based on P4, after applying the Godsil-McKay switching, the degree of each node in $X_i^1$ will be $d_w$ and  the degree of each node in $X_i^2$ will be $d_u$. Thus, by Rule 3, and the fact that $|X_i^1|=|X_i^2|$, the degree sequence of the graph does not change after switching. Moreover, it is easy to see that after the application of the switching to $X_i$, the graph still remains bipartite.

(iv) Let $X_i$ be a set of Type 3.2. We consider two cases:\\
{\bf Case 1.} Without loss of generality, assume that $d_w>d_u$. In this case, by P4, P5, and the condition $d_w>d_u$, there must be a  node  $y'\in Y$ such that  $y'$ is adjacent with $\frac{|X_i|}{2}$ nodes of $X_i^2$. Thus, $|X_i^2| \geq \frac{|X_i|}{2}$. This together with the definition of Type 3.2 sets, i.e.,  $|X_i^1| \neq |X_i^2|$, result in
\begin{equation}\label{V6E1}
|X_i^1|<|X_i^2|\:.
\end{equation}
Note that (\ref{V6E1}) is valid for any set $X_i$  of Type 3.2.
On the other hand, the set $X_i^1$ contains at least one node $v$. By P1 and P2, none of the $d_u$ connections of $v$ can be to any node in Type 1, Type 2 or Type 3.1 sets. The connections cannot be to the nodes in $X_i^2$ either, because this implies, by the condition of Godsil-McKay partitioning, that every node in $X_i$ must also be connected to $d_u$ other nodes in $X_i$. This however, is not possible because it would imply that there must be $|X_i^2| \times d_u$ connections from $X_i^2$ to $X_i^1$, which, by (\ref{V6E1}), is more than the total number of edges connected to $X_i^1$, i.e., $|X_i^1| \times d_u$. We thus conclude that there is at least a set $X_j$ of Type 3.2 such that each node of $X_i$ has at least one neighbor in $X_j$, and by P3, $|X_j^1|>|X_j^2|$.
But this contradicts (\ref{V6E1}). We thus come to the conclusion that this case cannot happen. \\
{\bf Case 2.} Now, assume that $d_w= d_u$. If there are two nodes $y$ and $y'$ in $Y$ such that $y$ is adjacent to
$\frac{|X_i|}{2}$ nodes of $X_i^1$ and $y'$ is adjacent to $\frac{|X_i|}{2}$ nodes of $X_i^2$, then $\frac{|X_i|}{2} \leq |X_i^1|$ and $\frac{|X_i|}{2} \leq |X_i^2|$. 
This implies $|X_i^1|=|X_i^2|=\dfrac{|X_i|}{2}$. But this contradicts the definition of Type 3.2 sets. 
Also, if there is a node $y\in Y$ such that $y$ is adjacent to $\frac{|X_i|}{2}$ nodes of $X_i^1$ (or $X_i^2$), but there is no node $y'\in Y$ such that $y'$ is adjacent to
$\frac{|X_i|}{2}$ nodes of $X_i^2$ (or $X_i^1$), then by P5, we have $d_u \neq d_w$, again a contradiction. Thus, there is no connection between the nodes in  $Y$ and those of $X_i$.
Let $S$ be the union of all Type-3.2 sets. Partition $S$ into two sets $S^1$ and $S^2$, where $S^1$ is a subset of $U$ and $S^2$ is a subset of $W$. 
Each node in $S$ has no neighbor in $Y$, Type-1, Type-2, or Type-3.1 sets. Now, consider the node-induced subgraph  on the set of nodes $S$. 
Since the degree of all nodes in $G$ are the same, by counting the number of edges from two sides, we have 
$|S^1| =|S^2|$. This combined with $d_u=d_w$, and Rule 3 shows that the Godsil-McKay switching does not change the degree sequence of $G$. The graph also remains bipartite. This completes the proof.
\end{proof}

\begin{rem}
	From the discussions above, one can see that a Godsil-McKay switching partition of bi-regular bipartite graphs, in which degrees of the two sides are unequal, cannot have Type 3.2 sets. Thus, for practical Tanner graphs in which $d_u \neq d_w$, 
	a valid Godsil-McKay switching partition $\{Y, X_1, \ldots, X_\ell\}$ of the nodes can only have $X_i$'s that are either Type 1, Type 2 or Type 3.1. There can also be connections only between Types 1 and 2, and between Types 3.1 and 3.1. 
	Nodes in $Y$ can be connected to the nodes in all three types of $X_i$ sets. 	
\end{rem}

It is important to note that, in general, Godsil-McKay switching does not preserve the degree sequences of a graph. 
As an example, consider the half-regular bipartite graph $G$ shown in Fig. \ref{V3G1}(a). Let $\ell=3$, and choose $X_1=\{v_1,v_2,v_3,v_4 \}$, $X_2=\{u_1,u_2,u_3,u_4 \}$, $X_3=\{z_1,z_2,z_3,z_4 \}$ and
$Y=\{ x_1,x_2,x_3 \}$. 
The partitioning has all necessary properties of Theorem \ref{Th5}. 
We can thus apply Godsil-McKay switching. By applying the switching, we obtain the graph $G'$, given in Fig. \ref{V3G1}(b). 
One can see that although $G'$ is also half-regular with the same degree of two on the regular side, the degree sequence of the two graphs differ on the irregular side.


\begin{figure}[ht]
\begin{center}
\includegraphics[scale=.5]{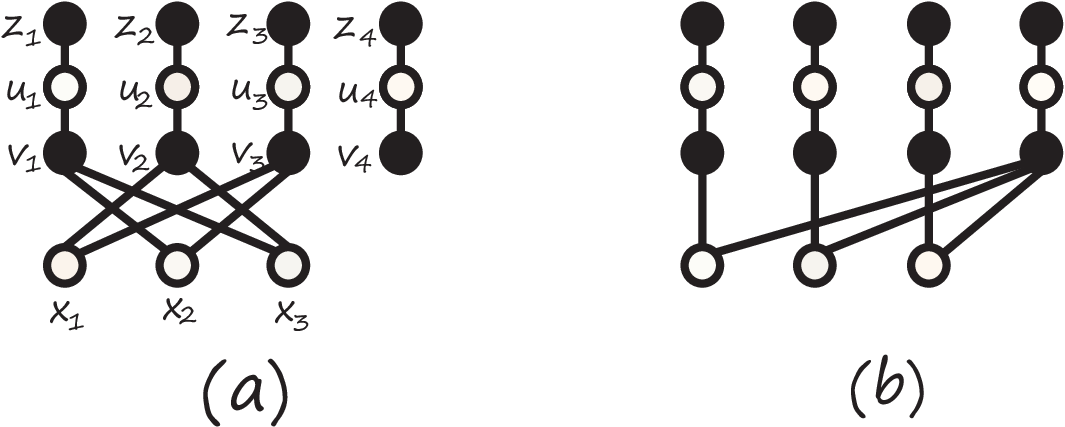}
\caption{Two cospectral graphs: (a) $G$ and (b) $G'$.}
\label{V3G1}
\end{center}
\end{figure}

\section{Computing the number of cycles in irregular bipartite graphs}
\label{sec35}

In this section, we consider the problem of counting the cycles of different length in irregular bipartite graphs of different girth $g$. First, we demonstrate through counter-examples 
that if $g=4$, the information of degree sequences and spectrum is, in general, insufficient to count $i$-cycles for any $i \geq g+2$. 
Next, for $g \geq 6$, we show by counter-examples that spectrum and degree sequences cannot, in general, uniquely determine the multiplicity of $i$-cycles for any $i \geq g$. 
The results for the case of $g \geq 8$ are provided before those of $g=6$,
since the graphs constructed for the former case are used as building blocks for graph constructions in the latter.

\subsection{$g=4$: Counter-example for $i$-cycles, $i \geq g+2$}
\label{bnh}

In this subsection, we construct two half-regular bipartite graphs such that they both have the same spectrum, degree sequence and girth $4$, but have different number of $i$-cycles for $i \geq 6$.

\underline{Construction of the graph $\mathcal{G}$}:

Consider two disjoint cycles of length $4$ and $12$ with node sets $\{v_1, v_2, v_3, v_4\}$, and $\{u_1, u_2, \ldots, u_{12}\}$, respectively. Add two nodes $w$ and $w'$ to the union of the cycles, and connect both $w$ and $w'$ to the nodes $v_1,v_3,u_1,u_3$. Also, add another node $w''$, and connect it to the nodes $u_5,u_7,u_9,u_{11}$. Finally, add two more nodes $x$ and $y$ to the graph, and connect them to the nodes $v_2,v_4,u_2,u_4,u_6,u_8,u_{10},u_{12}$. 
Call the resultant graph $\mathcal{G}$. The graph $\mathcal{G}$ is bipartite, and has $21$ nodes and its girth is $4$. See Fig. \ref{graphGG2}. Consider the node partition $V(G) = U\cup W$ for  $\mathcal{G}$, where
$W=\{x,y, v_1,v_3,u_1,u_3,u_5,u_7,u_9,u_{11}\}$. We thus have $n=|U|=11$ and $m=|W|=10$. The degree sequence of $W$ is $(8,8,4,4,4,4,3,3,3,3)$, and the degree of each node in $U$ is $4$. 
Thus, $\mathcal{G}$ is variable-regular with variable  degree $4$.

\begin{figure}[ht]
\begin{center}
\includegraphics[scale=.40]{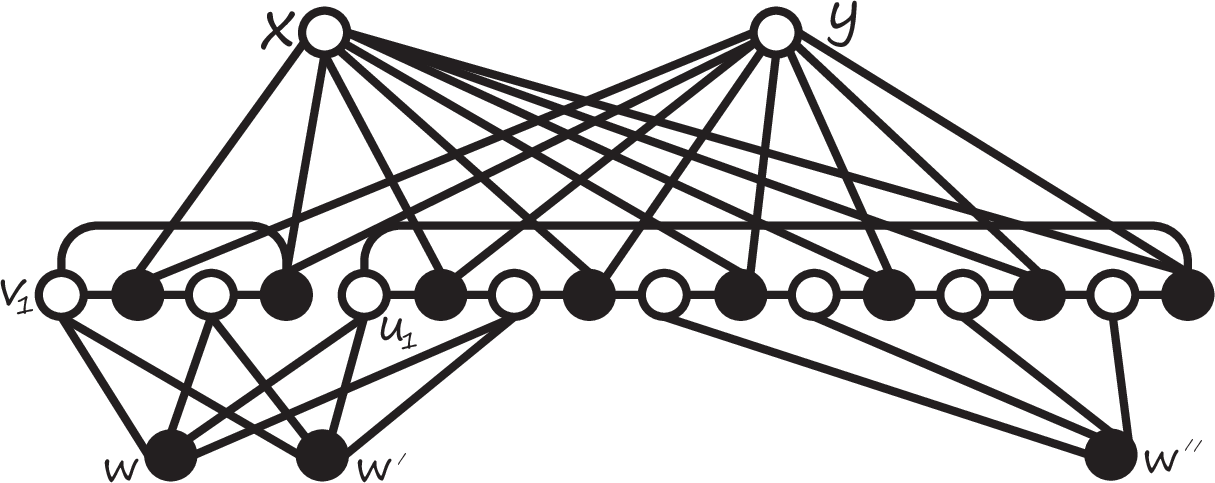}
\caption{Graph $\mathcal{G}$ of Subsection~\ref{bnh}.
} \label{graphGG2}
\end{center}
\end{figure}

\underline{Construction of  $\mathcal{G}'$ from  $\mathcal{G}$}: We use Godsil-McKay switching of Theorem \ref{Th5}. We choose $\ell=2$, $X_1=\{v_i,u_i: i \text{ is odd} \}$ and
$X_2=\{v_i,u_i: i \text{ is even} \}$. Thus, $|X_1|=|X_2|=8$. Also, we select $Y= \{w,w',w'',x,y\}$.  Nodes $w, w'$ and $w''$, each has $4$ neighbors in $X_1$, and no neighbor in $X_2$. 
Also, each of the nodes $x$ and $y$ has $8$ neighbors in $X_2$ and  no neighbor in $X_1$. Moreover, for each pair $i,j \in \{1,2\}$, all nodes in $ X_i$ have the same number of
neighbors in $X_j$. 
The partitioning has all necessary properties of Theorem \ref{Th5}, and thus, we can apply Godsil-McKay switching.
By applying the switching, we obtain the graph $\mathcal{G}'$, which has the same degree sequences as $\mathcal{G}$. In Table \ref{TT3}, we have listed the cycle distribution of both graphs for cycle lengths up to 18. 
One can see that $\mathcal{G}$ and $\mathcal{G}'$, although having the same spectrum, degree sequences and $g=4$, have different number of $i$-cycles for $i=6,8,\ldots,18$.


\begin{table}[ht]
\caption{Multiplicities of cycles of length $4$ up to $18$ in $\mathcal{G}$ and $\mathcal{G}'$, constructed in Subsection~\ref{bnh}}
\begin{center}
\scalebox{1}{
\begin{tabular}{ |c ||c|c|c|c|c|c|c|c|  }
\hline
Graph            &  4-cycles &  6-cycles &  8-cycles &  10-cycles &  12-cycles &  14-cycles &  16-cycles &  18-cycles \\ \hline
$\mathcal{G}$    &  60       &  248      &  1300      &  4056      &  11992     & 29780      & 43040     &  32640  \\ \hline
$\mathcal{G}'$   &  60       &  250      &  1294      &  4026      &  11706     & 28440      & 41656     &  32096\\
\hline
\end{tabular}
}
\end{center}
\label{TT3}
\end{table}

\subsection{$g \geq 8$: Counter-examples for $i$-cycles with $i \geq g$}
\label{subsec22}

In this subsection, we consider irregular bipartite graphs with girth $g$ at least eight, and demonstrate that the information of spectrum, and degree sequences is not sufficient, in general, to determine the multiplicity of $i$-cycles for 
$i \geq g$. For this, in the following, for each  $t \geq 1$,  we first construct two irregular bipartite graphs $G_t$ and $G_t'$ such that they have the same spectrum and degree sequences, but different number of $(6+2t)$-cycles (one vs. zero). The disjoint union of these graphs can then be used to provide counter-examples for cospectral irregular graphs with the same degree sequences and the same girth $g$ (for any girth $g \geq 8$), but with different number of $i$-cycles for any $i \geq g$. (We note that the irregular graph constructed by the disjoint union of $G_t$ graphs, $t \geq \tau$, has girth $6+2 \tau$, while the corresponding disjoint union of $G_t'$ graphs has an infinite girth. To make an example where both graphs have the same girth, one can simply consider the union of the constructed graphs with a cycle of length $6+2 \tau$.)

\begin{figure}[ht]
\begin{center}
\includegraphics[scale=.38]{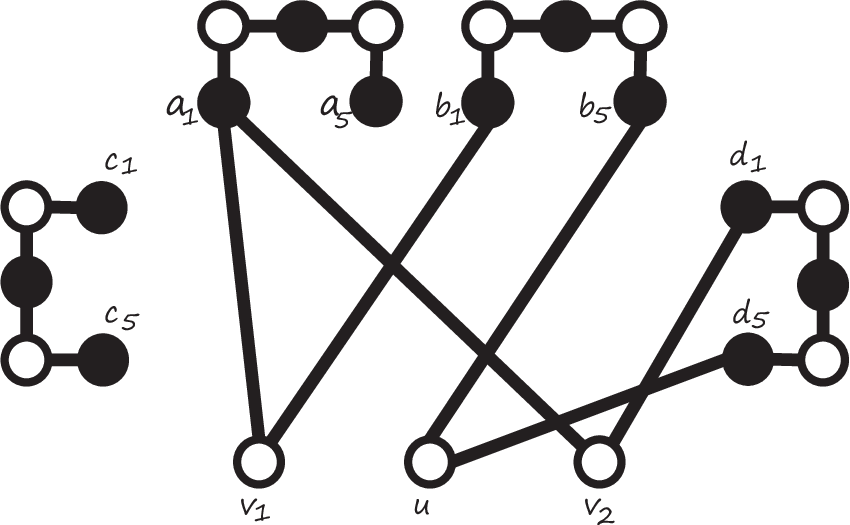}
\caption{Graph $G_4$ constructed in Subsection~\ref{subsec22}.
} \label{graphGG3}
\end{center}
\end{figure}

\underline{Construction of the graph $G_t$}:
Consider the integer $t\geq 1$, and four paths, each of length $t$, with the node sets $\{a_1, \ldots, a_{t+1}\}$, $\{b_1, \ldots, b_{t+1}\}$, $\{c_1, \ldots, c_{t+1}\}$, and $\{d_1, \ldots, d_{t+1}\}$, respectively. 
Then, add three nodes $v_1,v_2$ and $u$, and the edges $v_1a_1$, $v_1b_1$, $v_2a_1$, $v_2d_1$, $ud_{t+1}$ and $u b_{t+1}$, to the graph.
Call the resultant bipartite graph $G_t$.  As an example, the graph $G_4$ is shown in  Fig. \ref{graphGG3}. The graph $G_t$ has $4t+7$ nodes and 
only one cycle of length $6+2t$. From $4t+7$ nodes, $4t+3$ are of degree $2$, one node has degree $3$ and three nodes have degree $1$.

\underline{Constructing $G_t'$ from $G_t$}: We use Godsil-McKay switching of Theorem \ref{Th5} with $\ell=t+1$, and for each $i$, $1\leq i \leq t+1$, we select $X_i=\{ a_i,b_i,c_i,d_i\}$.
We thus have $Y= \{v_1,v_2,u\}$.  It can be seen that, for every node $y\in Y$, and every $i\in \{1,\ldots,t+1\}$, the node $y$ has either $0$  or $2$ neighbors
in $X_i$ (note that for each $i$, $|X_i|=4$). Also, for each pair $i,j\in \{1,\ldots, t+1\}$, all the nodes in $X_i$ have the same number of
neighbors in $X_j$. 
The $(i,j)$ entry of the following matrix shows the number of neighbors that an arbitrary node $v\in X_i$ has in the set $X_j$:
\[\begin{bmatrix}
0     &1     & 0   & 0  &    \cdots  &  0 &  0 &  0  \\
1     &0     & 1   & 0  &    \cdots  &  0 &  0 &  0  \\
0     &1     & 0   & 1  &    \cdots  &  0 &  0 &  0  \\
\vdots&\vdots&\vdots&\vdots&  \ddots &\vdots&\vdots&\vdots \\
0     &0     & 0   & 0  &    \cdots  &  1 &  0 &  1  \\
0     &0     & 0   & 0  &    \cdots  &  0 &  1 &  0
\end{bmatrix}\:.\]
This matrix  shows that for each pair $i,j\in \{1,\ldots,\ell\}$, all the nodes in $X_i$ have the same number of neighbors in $X_j$. 
Thus, the node partitioning has all the necessary properties of Theorem \ref{Th5} for the application of Godsil-McKay switching.
By applying the switching, we obtain the graph $G_t'$. For example, corresponding to $G_4$ in Fig.~\ref{graphGG3}, we obtain  the graph $G_4'$, shown in Fig. \ref{graphGG4}.
The graph $G_t'$ has the same spectrum and degree sequences as $G_t$, but does not have any cycle.

\begin{figure}[ht]
\begin{center}
\includegraphics[scale=.35]{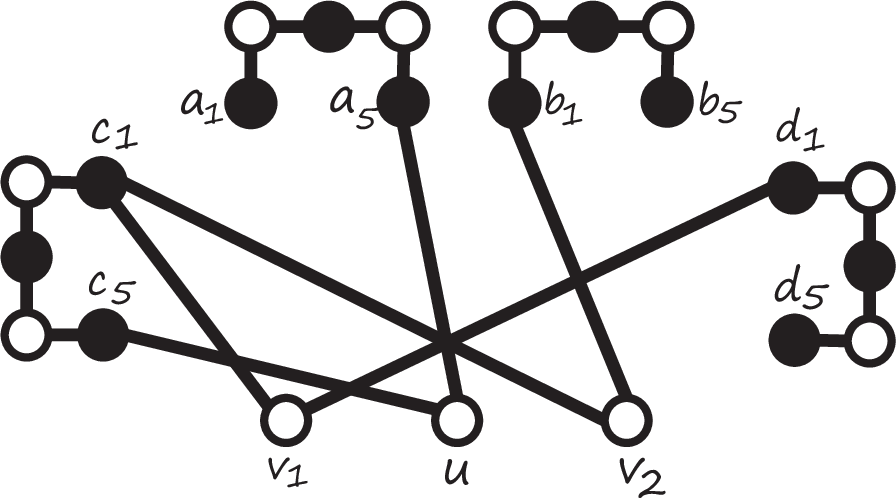}
\caption{Graph $G_4'$, obtained by Godsil-McKay switching from $G_4$, shown in Fig.~\ref{graphGG3}.
} \label{graphGG4}
\end{center}
\end{figure}

\begin{rem}\label{R2}
Consider the graph $G_t$, where $t$ is an even number. Consider the partition $V(G_t) = U\cup W$ for the nodes of $G_t$, where
$U=\{v_1,v_2,u \}\cup \{ a_i,b_i,c_i,d_i : i \text{ is even}\}$. The graphs $G_t$ and $G_t'$ are variable-regular with variable degree $2$. 
We thus conclude that the number of $g$-cycles in half-regular bipartite graphs cannot be, in general, computed using the spectrum and the degree sequences of the graph 
when the girth of the variable-regular bipartite graph is $g=6+2t$, where $t \geq 2$ is an even number.
\end{rem}

\subsection{$g =6$: Counter-examples for $i$-cycles with $i \geq g$}
\label{subsec33}

In this subsection, we first provide a counter-example of two cospectral irregular graphs with similar degree sequences and $g=6$, but different $N_6$. In Part \ref{subsub123}, we then construct two irregular bipartite graphs $G_{t,k}$ and $G_{t,k}'$ with girth $6$, such that they both have the same spectrum and degree sequences, but different multiplicity for $i$-cycles with $i \geq g+2$. 

\subsubsection{Counter-example for $6$-cycles}

Consider  the disjoint union of two $6$-cycles and two paths, each of length $5$, and call it $G_1$ (i.e., $G_1=2C_6 \cup 2P_6$). Also, consider the disjoint union of a $6$-cycle, a $14$-cycle and two paths, each of length one, and call it $G_2$ (i.e., $G_2=C_6 \cup C_{14} \cup 2 P_2$). It is easy to see that $G_1$ and $G_2$ are irregular bipartite graphs with the same degree sequences (both have ten nodes with degree $2$ and two nodes with degree $1$ on each side of the bipartition). Using (\ref{path}) and (\ref{cycle}), one can also see that $G_1$ and $G_2$  are cospectral. The girth of both graphs is six, but they have different number of $6$-cycles (two vs. one).

\subsubsection{Counter-example for $i$-cycles with $i \geq g+2$}
\label{subsub123}
\underline{Construction of the graph $G_{t,k}$}:
Let $t$ and $k$ be two integers such that $t > k \geq 0$, and $t+k$ is an even number. Consider the graph $G_t$ which was constructed in Subsection \ref{subsec22}. 
Add a path of length $k$ with the node set $\{f_1, \ldots, f_{k+1}\}$, as well as the edges $uf_1$ and $ v_2 f_{k+1}$ to $G_t$. Call the resultant graph $G_{t,k}$. As an example, the graph $G_{4,2}$ is shown in Fig. \ref{graphGG66}(a). 
The graph $G_{t,k}$ has $4t+k+8$ nodes, out of which, $4t+k+2$ nodes have degree $2$, three have degree $3$, and three have degree $1$. 
The graph is also bipartite and has one cycle of length $t+k+4$, one cycle of length $t+k+6$ and one cycle of length $6+2t$. 

\begin{figure}[ht]
\begin{center}
\includegraphics[scale=.35]{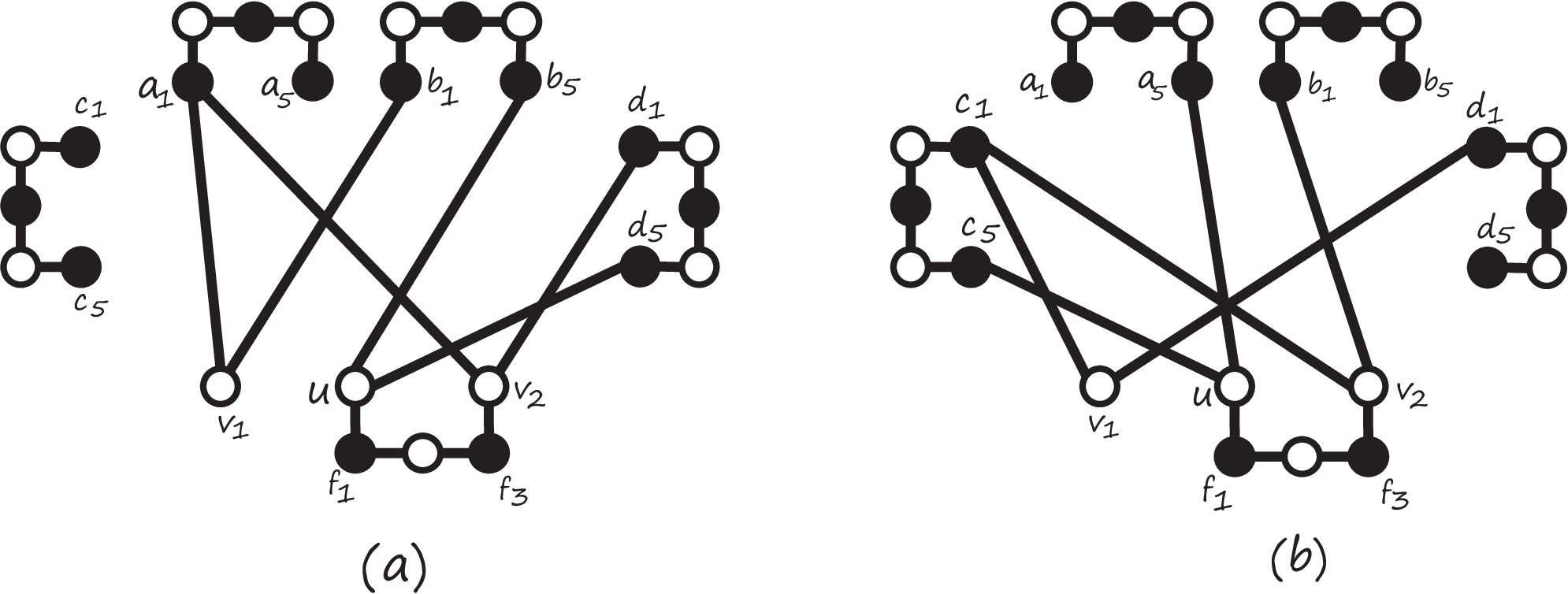}
\caption{Graphs (a) $G_{4,2}$ and (b) $G_{4,2}'$. 
} \label{graphGG66}
\end{center}
\end{figure}

\underline{Construction of $G_{t,k}'$ from $G_{t,k}$}: We use Godsil-McKay switching to transform $G_{t,k}$ into $G_{t,k}'$. Let $\ell=t+1$, and for each $i$, $1\leq i \leq t+1$, let $X_i=\{ a_i,b_i,c_i,d_i\}$.
We then have $Y= \{v_1,v_2,u,f_1,f_2,\ldots ,f_{k+1}\}$. All the conditions of Theorem~\ref{Th5} apply to this partition. We call the graph obtained by the switching $G_{t,k}'$.
The graph $G_{t,k}'$ can also be generated directly from $G_t'$, the Godsil-McKay switched version of $G_t$: 
Add to $G_t'$ a path of length $k$ with the node set $\{f_1, \ldots, f_{k+1}\}$, and the edges $uf_1$ and $ v_2 f_{k+1}$.  
As an example, in Fig. \ref{graphGG66}(b), the graph $G_{4,2}'$ is shown. The graph $G_{t,k}'$ has the same spectrum and degree sequences as $G_{t,k}$, but has only one cycle of length $t+k+4$.  
Now, for a fixed $i\geq 1$, consider the disjoint union of graphs $G_{2,0}, G_{3,1}, \ldots, G_{i+2,i}$, and call it $D_i$. Also, use  $D_i'$ to denote the disjoint union of graphs $G_{2,0}', G_{3,1}', \ldots, G_{i+2,i}'$.
For each $j \in \{0, \ldots, i\}$, the graph $G_{j+2,j}$ has one $(2j+6)$-cycle, one $(2j+8)$-cycle and one $(2j+10)$-cycle. Also, the graph $G_{j+2,j}'$ has only one cycle of length $2j+6$. 
Considering that the spectrum of a disconnected graph is the disjoint union  of the spectra of its components, one can see that $D_i$ and $D_i'$ are cospectral. They also have the same degree sequences and girth $g=6$.
It can however, be seen that while both graphs have only one cycle of length $6$, they have different number of $k$-cycles for each $6< k \leq 2i+10$. 
As an example, the cycle distributions of $D_3$ and $D_3'$ are given in Table \ref{V7T1}. 

\begin{table}[ht]
\caption{Multiplicities of cycles of length $6$ up to $16$ in Graphs $D_3$ and $D_3'$}
\begin{center}
\scalebox{1}{
\begin{tabular}{ |c ||c|c|c|c|c|c|c|   }
\hline
Graph    &     6-cycles &  8-cycles  &  10-cycles &  12-cycles &  14-cycles &  16-cycles \\ \hline
$D_3$    &     1        &  2         &  3         &  3         & 2          & 1      \\ \hline
$D_3'$   &     1        &  1         &  1         &  1         & 0          & 0             \\
\hline
\end{tabular}
}
\end{center}
\label{V7T1}
\end{table}

\section{Counting cycles in half-regular bipartite graphs}
\label{sec36}

The counter-example constructed in Subsection~\ref{bnh} for $g=4$ was based on half-regular bipartite graphs. We thus know that if $g=4$, the knowledge of spectrum and degree sequences is not sufficient in general to count the number of $i$-cycles for $i \geq g+2$ in half-regular bipartite graphs. On the other hand, the positive result of \cite{eigenvalue} is applicable to half-regular graphs and can be used to compute $N_4$.  
Furthermore, in Remark \ref{R2}, we showed that, in general, one cannot find $N_g$ for $g = 6+2t$, where $t \geq 2$ is an even number, in half-regular graphs 
just by using the information of spectrum and degree sequences. 
In this section, we complement these results.
We  present counter-examples  for $g$-cycles if $g=6+2t$, where $t \geq 1$ is an odd number, or for $i$-cycles with $i \geq g+2$, in graphs with $g \geq 6$.

\subsection{Counter-examples for $g$-cycles ($g=6+2t, t \geq 1$ and odd)}

Consider the disjoint union of two cycles, each of length $6+2t$, and two paths, each of length $5+t$, and call it $G_1$. Also, consider the disjoint union of a $(6+2t)$-cycle, a $(14+2t)$-cycle and two paths, each of length $t+1$, and call it $G_2$. 
One can see that both $G_1$ and $G_2$ are half-regular bipartite graphs and have the same degree sequences (the regular side has $11+3t$ degree-$2$ nodes and the irregular side has $9+3t$ degree-$2$ and $4$ degree-$1$ nodes).\footnote{Note that if $t$ is selected to be an even number, the graphs $G_1$ and $G_2$ will not be half-regular.}
Using (\ref{path}) and (\ref{cycle}), one can also see that $G_1$ and $G_2$ are cospectral, and both have girth $g=6+2t$. The number of $g$-cycles $N_g$, however, is different for each graph (two vs. one).  

\subsection{$g \geq 6$: Counter-examples for $i$-cycles, $i \geq g+2$}

In this subsection, we construct variable-regular bipartite graphs that have the same spectrum, degree sequences and girth $g \geq 6$, but have different multiplicities of $i$-cycles for $i \geq g+2$. 
We first start by constructing two graphs $\mathcal{G}_{t,k}$ and $\mathcal{G}_{t,k}'$, related by Godsil-McKay switching.

\underline{Construction of the graph $\mathcal{G}_{t,k}$}:
Let $t$ and $k$ be two even integers such that $t \geq k\geq 0$ and $t>0$. Consider the graph $G_{t,k}$ which was constructed in Subsection \ref{subsec33}. For each node $z$ in the set $\{a_i,b_i,c_i,d_i,f_i: i \:\:\:\text{even}\}\cup \{v_1\}$, add 
a new node $z'$ to $G_{t,k}$, and connect $z$ to $z'$.  Call the resultant variable-regular graph $\mathcal{G}_{t,k}$. As an example, Fig. \ref{graphGG88}(a) shows $\mathcal{G}_{4,2}$. 
The graph $\mathcal{G}_{t,k}$ is bipartite and has one cycle of length $t+k+4$, one cycle of length $t+k+6$ and one cycle of length $6+2t$.
(Note that if $t=k$, then the graph $\mathcal{G}_{t,k}$ has one cycle of length $2t+4$, and two cycles of length $2t+6$.) 

\underline{Construction of the graph $\mathcal{G}_{t,k}'$}:
We use Godsil-McKay switching of Theorem \ref{Th5} to construct $\mathcal{G}_{t,k}'$ from $\mathcal{G}_{t,k}$. 
Let $\ell=3t/2+1$, and for each $i$, $1\leq i \leq t+1$, let $X_i=\{ a_i,b_i,c_i,d_i\}$. Also, for each $i$, $t+2\leq i \leq 3t/2+1$, let $j=2(i-t-1)$, and $X_i=\{ a_j',b_j',c_j',d_j'\}$.
We thus have $Y= \{v_1,v_2,u,f_1,f_2,\ldots ,f_{k+1}\}\cup  \{f_i': i \: \: \:\text{even}\} \cup \{v_1'\}$. It can be seen that this partitioning satisfies all the required conditions of Theorem \ref{Th5}.
We thus apply the switching and obtain the graph $\mathcal{G}_{t,k}'$. The graph $\mathcal{G}_{t,k}'$ can also be constructed by the following approach: 
Consider the graph $G_{t,k}'$ which was constructed in Subsection \ref{subsec22}. For each node $z$ in the set $\{a_i,b_i,c_i,d_i,f_i: i \text{even}\}\cup \{v_1\}$, add a new node $z'$ to $G_{t,k}'$, and connect $z$ to $z'$. 
As an example, Fig. \ref{graphGG88}(b) shows the graph $\mathcal{G}_{4,2}'$. The graph $\mathcal{G}_{t,k}'$ is also variable-regular bipartite and has the same spectrum and degree sequences as $\mathcal{G}_{t,k}$. 
It however, has only one cycle of length $t+k+4$.  

\begin{figure}[ht]
\begin{center}
\includegraphics[scale=.30]{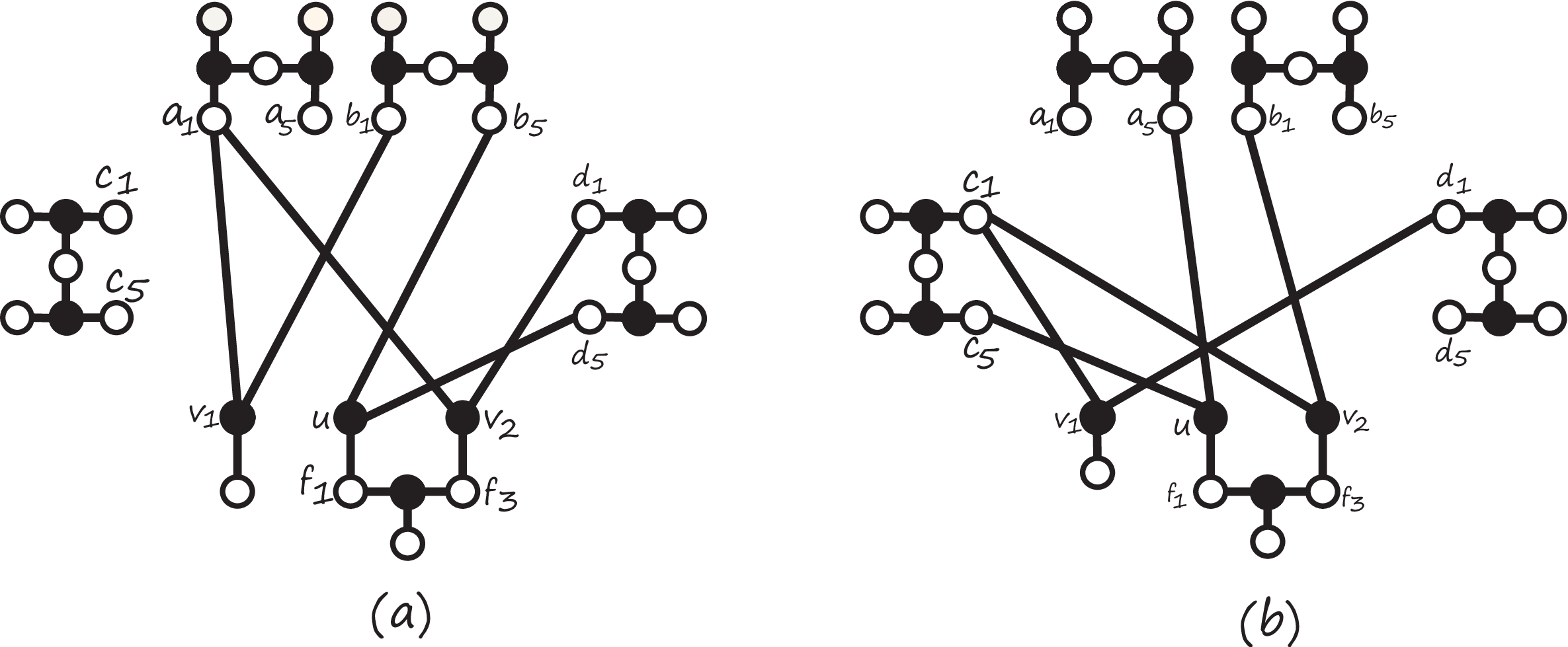}
\caption{Graphs (a) $\mathcal{G}_{4,2}$ and (b) $\mathcal{G}_{4,2}'$. 
} \label{graphGG88}
\end{center}
\end{figure}

Let $i$ be an even number. The graph $\mathcal{G}_{i+2,i}$ has one $(2i+6)$-cycle, one $(2i+8)$-cycle and one $(2i+10)$-cycle. The graph
$\mathcal{G}_{i+2,i}'$, however, has only one cycle of length $2i+6$. 
Now, for fixed integers $j$ and $k$ satisfying $j \geq k \geq 1$, consider the disjoint union of graphs $\mathcal{G}_{2k,2k-2}, \mathcal{G}_{2k+2,2k}, \ldots, \mathcal{G}_{2j,2j-2}$, and call it ${\cal F}_{j,k}$.  Also, consider the disjoint union of graphs 
$\mathcal{G}_{2k,2k-2}', \mathcal{G}_{2k+2,2k}', \ldots, \mathcal{G}_{2j,2j-2}'$, and call it ${\cal F}_{j,k}'$. Both ${\cal F}_{j,k}$ and ${\cal F}_{j,k}'$ have the same spectrum and degree sequences. They also have the same girth of $4k+2$, and both have one $(4k+2)$-cycle. They however, have different number of $\ell$-cycles for any $4k+2 < \ell \leq 4j+6$.

To cover the cases where $g=4(k+1), k \geq 1$, let $k'$ be an odd number satisfying $k' > 2k+1$, and consider two graphs $G_1$ and $G_2$, where $G_1$ is the disjoint union of the cycle $C_{4(k+1)}$ and two copies of the
path $P_{k'}$, and $G_2$ is the disjoint union of $C_{2(k'+1)}$, and two copies of the path $P_{2k+1}$. One can see that $G_1$ and $G_2$ are half-regular bipartite graphs with similar degree sequences. It can also be seen, using (\ref{path}) and (\ref{cycle}), that both graphs have the same spectrum. The two graphs, however, have different cycle distributions, i.e., while $G_1$ has one cycle of length $4(k+1)$, $G_2$ has one cycle of larger length $2(k'+1)$. Now, if one considers the disjoint unions of $G_1$ and $G_2$ with a cycle of length $4(k+1)$, then the resultant graphs both have the same girth of $4(k+1)$, but they have different number of cycles of length $2(k'+1)$.

\section{Concluding remarks}
\label{sec5}

It is well-known that the number of closed walks in a graph can be computed using the spectrum of the graph. It is also known that the multiplicity 
of cycles of length larger than three cannot be determined only by the knowledge of the spectrum. Recently, in~\cite{blake2017short},~\cite{eigenvalue}, 
it was shown that adding the knowledge of degree sequences for bipartite graphs to the information about the spectrum will enable the computation of multiplicities of cycles of certain lengths
in bi-regular, half-regular and irregular graphs. (See Table~\ref{TXY1}.)
In this work, we complemented the results of~\cite{blake2017short, eigenvalue}, and demonstrated, by constructing counter-examples, that for the remaining cases,
the information of the spectrum and degree sequences is insufficient, in general, to determine the multiplicity of cycles. An interesting topic of research would be to 
determine what extra information, in addition to degree sequences and spectrum, is required to compute the multiplicity of cycles of length larger than or equal to $2g$.

In Theorem \ref{V2Th2}, we proved that Godsil-McKay switching preserves the degree sequences of bi-regular bipartite graphs. It is also known that the spectrum of the graph is preserved under this switching. This implies that the multiplicity of short cycles of length up to $2g-2$ remains unchanged with the application of the Godsil-McKay switching to a bi-regular bipartite graph. On the other hand, short cycles and their combinations form graphical objects that trap the iterative decoding algorithms of LDPC codes. An interesting topic would be to study the effect of Godsil-McKay switchings on the distribution of trapping sets and the possibility of reducing the multiplicity of trapping sets through the application of this switching.

\bibliographystyle{plain}

\end{document}